\documentclass[12pt,letterpaper]{article}

\usepackage[margin=1in]{geometry}
\usepackage{amsmath,amssymb}
\usepackage{amsthm}
\usepackage{cite}           
\usepackage{nameref,hyperref}
\usepackage[right]{lineno}
\usepackage[nopatch=eqnum]{microtype}
\DisableLigatures[f]{encoding = *, family = *}
\usepackage[table]{xcolor}
\usepackage{array}
\usepackage{booktabs}
\usepackage{enumitem}
\usepackage{mathtools}
\usepackage{setspace}
\usepackage{lastpage,fancyhdr,graphicx}

\makeatletter
\renewcommand{\@biblabel}[1]{\quad#1.}
\makeatother

\raggedright
\usepackage[aboveskip=1pt,labelfont=bf,labelsep=period,
            justification=raggedright,singlelinecheck=off]{caption}

\newtheorem{proposition}{Proposition}

\newtheorem{assumption}{Assumption}

\newcommand{\indicator}{\mathbf{1}}
\newcommand{\xG}{x_G}
\newcommand{\xP}{x_P}
\newcommand{\xR}{x_R}
\newcommand{\fG}{f_G}
\newcommand{\fP}{f_P}
\newcommand{\fR}{f_R}
\newcommand{\fbar}{\bar{f}}
\newcommand{\eeff}{e}
\newcommand{\estar}{e^*}
\newcommand{\cG}{c_G}
\newcommand{\cP}{c_P}
\newcommand{\bG}{b_G}
\newcommand{\bP}{b_P}

\begin{document}

\vspace*{0.2in}

\begin{flushleft}
{\Large
\textbf{The partial adoption trap: Coordination failure, trust, and
cultural lock-in in health AI adoption}
}
\newline
\\
Ari Ercole\textsuperscript{1*}
\\
\bigskip
\textbf{1} Cambridge Centre for AI in Medicine and Department of
Medicine, University of Cambridge, Cambridge, United Kingdom
\\
\bigskip
* \href{mailto:ae105@cam.ac.uk}{ae105@cam.ac.uk}
\\
\bigskip
ORCID: \href{https://orcid.org/0000-0001-8350-8093}{0000-0001-8350-8093}
\end{flushleft}

\section*{Abstract}

Health artificial intelligence (AI) adoption presents a paradox:
point-solution tools diffuse readily through clinical populations, yet
system-change AI, which carries the greatest potential for pathway-level
transformation, consistently stalls at partial adoption. An evolutionary
game theoretic model is developed to explain this pattern. Doctors choose
among three strategies: genuine adoption, partial adoption, and
rejection, where genuine adoption is required for systemic benefits to
materialise above a population threshold. The system is shown to be
generically bistable, with a stable partial adoption equilibrium
coexisting alongside full genuine adoption. The basin of attraction of
the partial adoption trap is enlarged by three compounding failure modes:
a threshold coordination failure arising from the non-appropriable nature
of systemic benefits; a trust failure arising from the organisation's
inability to credibly commit to sharing productivity gains; and a cultural
failure arising from negative coordination norms among doctors. These
failure modes are shown to be most severe precisely for the technologies
with the greatest systemic value: the Value-Adoption Paradox. A cost
ratchet dynamic implies that failed adoption attempts permanently lower
barriers even when embedding fails, but this benefit is offset when
trust erosion is rapid. Conditions are derived under which sustained but
imperfect adoption pressure is welfare-improving, and the policy
architecture required to escape the trap (targeting trust, sequencing,
and team-level adoption) is characterised. Standard health system
digital transformation policy, which typically addresses only the
threshold failure through individual incentives, is predicted to
systematically produce the partial adoption trap.

\clearpage

\section*{Introduction}
\label{sec:intro}

The deployment of artificial intelligence in health systems has
accelerated markedly over the past decade. Diagnostic support tools,
administrative automation, and clinical decision aids have proliferated,
and adoption rates by conventional measures appear substantial. Yet a
striking divergence has emerged between adoption of point-solution
AI (tools that improve performance on a specific clinical task for
the adopting clinician) and system-change AI, which is intended to
restructure care pathways, reduce transaction costs, and generate
productivity gains at the system level. Point solutions diffuse readily.
System-change AI is widely installed but rarely embedded: the tool is
present, often formally adopted, but the pathway transformation it is
intended to enable does not occur.

Ambient voice technology (AVT), AI scribes that transcribe
consultations and automate clinical documentation in real time, illustrates the problem acutely. AVT is explicitly identified as a
core enabler of the analogue-to-digital shift in the NHS
\textit{Fit for the Future} 10-Year Health Plan \cite{NHS10YP2025},
which commits to national rollout across all care settings and targets
the productivity gains needed to meet a 2\% annual efficiency
improvement. The individual benefit is real and well-documented:
documented time savings allow each doctor to do more with existing
capacity \cite{NHSEnglandAVT2025}. Yet widespread partial adoption
of AVT, using the tool to ease documentation without restructuring
clinic flow or throughput, will not deliver the capacity gains the
Plan anticipates. That requires enough of a clinic population to
genuinely restructure appointment patterns that the booking system,
administrative staff, and patient pathways all change coherently: a
collective threshold that no individual doctor can cross alone. AVT
is used throughout this paper as a running example because its
structure maps cleanly onto the general model, but the argument
applies wherever system-change AI requires collective workflow
restructuring to generate its intended benefit.

This divergence between individual and collective adoption incentives
is not explained by standard technology adoption models.
Diffusion-of-innovations theory \cite{Rogers1962} treats adoption as
binary and individual, driven by perceived usefulness and ease of use.
Technology acceptance models \cite{Davis1989} similarly focus on
individual-level attitudes. Neither framework captures the collective
action structure that characterises system-change AI, where the benefits
to any individual clinician depend on the adoption decisions of
colleagues, and where systemic benefits only materialise when adoption
exceeds a population threshold.

The economics literature offers a closer analogue. Canton et al.
\cite{Canton1999} develop a model in which workers resist technology
adoption when costs fall disproportionately on those with least time
to recoup benefits, and where the share of productivity gains accruing
to workers depends on market competition. Their overlapping-generations
framework with majority voting captures the distributional politics of
adoption resistance. However, that model does not capture three features
distinctive to clinical AI adoption in a public health system: the
threshold structure of systemic benefits, the trust problem arising
from the organisation's inability to credibly commit to sharing gains,
and the social dynamics by which partial adoption norms become
self-reinforcing within clinical teams.

This paper contributes to the literature by developing an evolutionary
game theoretic (EGT) model of clinical AI adoption that nests all three
failure modes in a single framework and characterises their interactions.
EGT is appropriate here for two reasons. First, adoption decisions in
large clinical populations are made by boundedly rational agents who
observe and imitate peers rather than solving complex optimisation
problems, a pattern well-documented in the clinical practice variation
literature \cite{Wennberg2010} and consistent with the imitation
dynamics that underpin replicator equations. Second, EGT naturally
accommodates a population-level strategy distribution as the state
variable, allowing characterisation not just of equilibria but of basins
of attraction and the dynamics of convergence, which are essential for
policy design.

The model has three strategies. Genuine adoption (G): the doctor
integrates the AI system in a way that restructures fundamental working
patterns, contributing to the population threshold required for systemic
benefit. Partial adoption (P): the doctor uses the tool for tasks where
it generates immediate private benefit (time savings, decision
support) but does not restructure working patterns, capturing private
gains without contributing to systemic transformation. Rejection (R):
no adoption. This strategy space is richer than existing models and
captures a distinction that is both empirically important and analytically
consequential: partial adoption is not non-adoption, but it is not the
adoption that generates systemic value.

The main results are as follows. First, under plausible parameter
conditions the replicator dynamics on the strategy simplex are bistable,
with two stable equilibria: full genuine adoption and the partial
adoption trap in which the entire population partially adopts. The
partial adoption trap is the generic attractor for most initial
conditions. Second, the cost of genuine adoption decays over time
conditional on the system embedding (a cost ratchet)which creates
path dependence: failed adoption attempts that push effective adoption
above the systemic threshold, even temporarily, permanently lower future
barriers. Third, when the organisation's sharing of productivity gains
is modelled as a trust game, rational anticipation of reneging
endogenises the appropriability parameter and generates a
self-confirming trust trap that compounds the coordination failure.
Fourth, negative coordination norms among doctors create cultural
lock-in that deepens the partial adoption trap independently of
individual incentives. Fifth, these failure modes compound each other
most severely for technologies with the greatest systemic value, the
Value-Adoption Paradox, implying that standard adoption policy fails
most where it matters most.

The policy implications are concrete. Escaping the partial adoption trap
requires addressing all four failure modes in the correct sequence. Trust
architecture must precede adoption requirements. Cultural preparation
must precede individual incentives. Threshold subsidies must be
concentrated and targeted at seeding rather than universal adoption.
Implementation support must be concentrated in the embedding window when
cost decay occurs.

The remainder of the paper is organised as follows. The next section
develops the model. The equilibrium analysis and phase portrait are then
presented, followed by the cost ratchet dynamics, the trust game, and
doctor-doctor coordination effects. The technology type index and the
Value-Adoption Paradox are formalised, followed by welfare analysis and
policy implications. The final section concludes.

\section*{Materials and methods}
\label{sec:model}

\subsection*{Setting and players}

Consider a large population of doctors employed within a health system
that is introducing a system-change AI tool. The tool is intended to
restructure a clinical pathway in a way that generates systemic
productivity gains, reduced waiting times, fewer duplicate
investigations, lower transaction costs across the pathway, but these
gains only materialise when a sufficient fraction of the population has
genuinely restructured their working practice around the tool. Individual
doctors choose how to engage with the tool, and the population
distribution of strategies determines whether systemic benefits arise.

Population dynamics are modelled using evolutionary game theory. The
justification is twofold. First, clinical adoption decisions are heavily
peer-influenced: doctors observe colleagues' behaviour and revise their
own practices accordingly, consistent with imitation dynamics
\cite{Szabo2007}. Second, the population is large enough that individual
optimisation against a fixed environment is a reasonable approximation,
with the environment itself determined by the aggregate strategy
distribution. The replicator dynamics that govern strategy frequencies
in EGT capture both features.

Alongside the doctor population, a second actor, the organisation
(hospital trust, integrated care system, or equivalent), sets the
implementation context and determines the fraction of productivity gains
shared with doctors. The organisation's behaviour is modelled through a
trust game below; the sharing fraction is treated as a parameter until
then.

\subsection*{Strategy space}

Doctors choose among three strategies:

\begin{description}[leftmargin=2em]
  \item[Genuine adoption (G):] The doctor integrates the AI system into
    practice in a way that restructures fundamental working patterns.
    This includes changing triage protocols, referral pathways, or
    documentation practices in the manner the tool is designed to enable.
    Genuine adoption bears a disruption cost and contributes fully to
    the population threshold for systemic benefit.

  \item[Partial adoption (P):] The doctor uses the tool for tasks where
    it generates immediate private benefit, time savings on specific
    administrative tasks, decision support for individual clinical
    decisions, but does not change fundamental working patterns. Time
    saved through partial use is reinvested in existing workflow
    priorities rather than the pathway restructuring the tool is designed
    to enable. Partial adoption bears a moderate cost, generates direct
    private benefit, and contributes partially to the population threshold.

  \item[Rejection (R):] The doctor does not adopt the tool. No cost,
    no benefit, no contribution to the threshold.
\end{description}

This strategy space captures the empirically important phenomenon of
partial adoption, which is widely documented in the health IT literature
\cite{Patterson2018,Blijleven2022} but typically treated as a measurement
issue rather than a strategic equilibrium. The analysis here treats it
as the latter.

\subsection*{Payoff functions}

Let $\xG$, $\xP$, $\xR \in [0,1]$ denote the population frequencies
of genuine adopters, partial adopters, and rejecters respectively, with
$\xG + \xP + \xR = 1$. Denote the state vector
$\boldsymbol{x} = (\xG, \xP, \xR)$.

\subsubsection*{Effective adoption and the threshold benefit}

The \emph{effective adoption level} is defined as:
\begin{equation}
  \eeff(\boldsymbol{x}) = \xG + \gamma \xP, \quad \gamma \in [0,1),
  \label{eq:effective_adoption}
\end{equation}
where $\gamma$ captures the partial contribution of partial adopters
to the systemic benefit. Partial adopters use the tool but do not
restructure pathways, so $\gamma < 1$; $\gamma$ is treated as a
technology and implementation parameter.

The systemic benefit function is:
\begin{equation}
  \Phi(\eeff) = B \cdot \sigma(\eeff - \estar),
  \label{eq:threshold_benefit}
\end{equation}
where $B > 0$ is the total systemic benefit available,
$\estar \in (0,1)$ is the adoption threshold, and $\sigma(\cdot)$ is
a smooth sigmoid function satisfying $\sigma(z) \to 0$ as
$z \to -\infty$ and $\sigma(z) \to 1$ as $z \to \infty$, with
$\sigma(0) = 1/2$ and $\sigma'(z) > 0$. The sigmoid approximates
the step function while preserving differentiability. Systemic
benefits are negligible below $\estar$ and approach $B$ above it.

\subsubsection*{Payoffs}

Let $\alpha \in (0,1)$ denote the fraction of systemic benefits that
each doctor appropriates, treated as a parameter until it is endogenised
in the trust game below.

\begin{assumption}[Cost and benefit ordering]
\label{ass:ordering}
The following inequalities hold: $c^{(R)} = 0 < \cP < \cG$; $\bG < \bP$;
and $\bP > \cP$, where $\cG$ is the disruption cost of genuine adoption,
$\cP$ is the cost of partial adoption, $\bG$ is the direct private
benefit of genuine adoption, and $\bP$ is the direct private benefit
of partial adoption.
\end{assumption}

The cost ordering reflects the greater workflow disruption required for
genuine adoption. The benefit ordering reflects the greater immediate
private utility of partial adoption. The final condition ensures that
partial adoption is individually profitable in isolation.

\begin{assumption}[Reputational cost function]
\label{ass:kappa}
The reputational cost function
$\kappa : [0,\hat\alpha] \to \mathbb{R}_{\geq 0}$ is strictly convex
and twice continuously differentiable ($C^2$), with $\kappa(0) = 0$,
$\kappa'(0) = 0$, and $\kappa''(z) > 0$ for all $z \in (0, \hat\alpha]$.
\end{assumption}

Strict convexity ensures $\kappa'$ is strictly increasing and hence
invertible, which is required for the organisation's first-order
condition to yield a unique optimal reneging level.

The payoff functions are:
\begin{align}
  \fG(\boldsymbol{x}, c, \alpha) &= -c(\eeff, t) + \alpha\Phi(\eeff) + \bG,
  \label{eq:payoff_G} \\
  \fP(\boldsymbol{x}) &= -\cP + \bP,
  \label{eq:payoff_P} \\
  \fR &= 0,
  \label{eq:payoff_R}
\end{align}
where $c(\eeff, t)$ is the time- and state-dependent disruption cost
of genuine adoption. The rejection payoff is normalised to zero without
loss of generality. Note that $\fP$ does not contain $\Phi(\eeff)$:
partial adopters do not restructure their workflow to capture systemic
returns. The mean population fitness is:
\begin{equation}
  \fbar(\boldsymbol{x}, c, \alpha) = \xG \fG + \xP \fP + \xR \fR.
  \label{eq:mean_fitness}
\end{equation}

\subsection*{Replicator dynamics}

Strategy frequencies evolve according to the standard replicator
equations:
\begin{align}
  \dot{\xG} &= \xG\bigl(\fG(\boldsymbol{x}, c, \alpha)
               - \fbar(\boldsymbol{x}, c, \alpha)\bigr),
  \label{eq:rep_G} \\
  \dot{\xP} &= \xP\bigl(\fP - \fbar(\boldsymbol{x}, c, \alpha)\bigr),
  \label{eq:rep_P} \\
  \dot{\xR} &= \xR\bigl(\fR - \fbar(\boldsymbol{x}, c, \alpha)\bigr),
  \label{eq:rep_R}
\end{align}
with $\dot{\xG} + \dot{\xP} + \dot{\xR} = 0$, so the dynamics remain
on the simplex
$S^2 = \{(\xG, \xP, \xR) : \xG + \xP + \xR = 1,\,
         \xG, \xP, \xR \geq 0\}$.

The replicator equation captures imitation dynamics: strategies that
perform above average increase in frequency, strategies that perform
below average decline. In clinical populations this corresponds to
doctors observing colleagues' experience with the tool and revising
their own approach accordingly, a mechanism well-supported in the
literature on peer effects in clinical practice
\cite{Avorn1982,Nembhard2009}.

\subsection*{The cost decay function}

The disruption cost of genuine adoption $c(\eeff, t)$ is both
time-varying and state-dependent, and its specification requires care
because the nature of the cost determines the form of its decay.

There are two conceptually distinct components. The first is an
individual learning cost: the friction a doctor experiences while
becoming proficient with a new tool, which diminishes through
personal practice regardless of what colleagues do. For AVT, this
component is genuinely small. Doctors typically adapt to the
transcription interface within a few sessions; the tool is designed
for non-technical users and the individual learning curve is shallow.
The second component is an environmental cost: the friction of
operating a restructured workflow in a system that has not yet
restructured around it. For a doctor who has genuinely adopted AVT
and committed to seeing additional patients, this is the dominant
source of disruption. The booking system still runs on the old
template. Administrative staff manage the changed appointment pattern
as an exception. Colleagues operating unchanged schedules create
coordination friction in shared spaces and support flows. None of
these frictions diminish through individual practice; they diminish
only when enough of the clinic has restructured that the new model
becomes the operating norm.

This distinction motivates a specification in which cost reduction
is conditional on population embedding rather than individual
practice time. The cost function is specified as:
\begin{equation}
  c(\eeff, t) = c_0 \exp(-\delta t) \cdot
                \indicator_{\{\eeff > \estar\}}
              + c_0 \cdot \indicator_{\{\eeff \leq \estar\}},
  \label{eq:cost_decay}
\end{equation}
where $c_0 > 0$ is the initial disruption cost, $\delta > 0$ is the
embedding rate, and $\indicator_{\{\cdot\}}$ is the indicator function.
The cost decays exponentially at rate $\delta$ when effective adoption
exceeds the threshold, reflecting that once a sufficient fraction of
the clinic has genuinely restructured, the new workflow becomes normal,
friction dissolves, and the old ways are gradually forgotten. If
$\eeff$ never crosses $\estar$, the cost remains at $c_0$
permanently: the environmental frictions persist because the system
never restructures.

The assumption that cost is exactly constant below threshold (rather
than decaying slowly) is a modelling simplification that warrants
explicit acknowledgement. In general, individual learning provides
some cost reduction even below threshold, and so a more complete
specification would include a small positive decay rate $\delta_{\text{ind}}$
below threshold alongside the larger embedding-driven rate
$\delta_{\text{emb}}$ above it. For system-change AI, however, the
individual component is small relative to the environmental one.
For AVT specifically, a doctor who has learned the tool but works in
an unrestructured clinic still faces the full environmental friction,
so $\delta_{\text{ind}} \approx 0$ is a reasonable approximation.
The key results of the paper carry through for the more general form
$\dot{c} = -\bigl(\delta_{\text{ind}} +
\delta_{\text{emb}} \cdot \indicator_{\{\eeff > \estar\}}\bigr)c$
with $\delta_{\text{ind}} \geq 0$ small (see Supporting Information);
the binary specification is adopted for tractability and transparency.

The cost does not recover if $\eeff$ falls back below $\estar$ after
having crossed it, giving:
\begin{equation}
  \dot{c} = -\delta c \cdot \indicator_{\{\eeff(\boldsymbol{x}) > \estar\}}.
  \label{eq:cost_ode}
\end{equation}
This irreversibility reflects the forgetting of old workflows once
new ones have been embedded: a clinic that restructures and then
reverts does not immediately recover the frictions it had overcome.
The irreversibility is the source of the cost ratchet result derived
below.

\section*{Results}
\label{sec:results}

\subsection*{Equilibrium analysis}

The dynamics of Eq~(\ref{eq:rep_G})--Eq~(\ref{eq:rep_R}) are first
analysed for fixed $c = c_0$ and fixed $\alpha$, treating the cost
decay and trust game as perturbations to be introduced below.
Throughout, Assumption~\ref{ass:ordering} is maintained.

\subsubsection*{Corner equilibria}

\begin{proposition}[Corner stability]
\label{prop:corners}
Under Assumption~\ref{ass:ordering}:
\begin{enumerate}[label=(\roman*)]
  \item The all-rejection corner $(0,0,1)$ is a saddle point, stable
    along the G--P edge (for $\bG < c_0$ and $\alpha B$ small) but
    unstable along the P--R edge whenever $\bP > \cP$.
  \item The partial adoption corner $(0,1,0)$ is locally asymptotically
    stable if $\bP > \cP$ and $\eeff = \gamma < \estar$.
  \item The genuine adoption corner $(1,0,0)$ is locally asymptotically
    stable if $\alpha B + \bG - c_0 > \bP - \cP$ and
    $\alpha B + \bG > c_0$.
\end{enumerate}
\end{proposition}

Standard linearisation of the replicator dynamics around each corner
yields Proposition~\ref{prop:corners} (see Supporting Information for
full proof). The saddle point structure of the rejection corner implies
that the relevant long-run competition is between genuine and partial
adoption: rejection is unstable against partial adoption whenever
partial adoption offers positive private returns.

\subsubsection*{Bistability and the partial adoption trap}

\begin{proposition}[Bistability]
\label{prop:bistable}
Under Assumption~\ref{ass:ordering} and the condition
$\alpha B > (c_0 - \cP) + (\bP - \bG)$, the replicator dynamics on
$S^2$ are bistable: both $(0,1,0)$ and $(1,0,0)$ are locally
asymptotically stable.
\end{proposition}

\begin{proposition}[Separatrix existence]
\label{prop:separatrix}
Under the conditions of Proposition~\ref{prop:bistable}, there exists
a separatrix $\mathcal{S} \subset S^2$ dividing the simplex into two
basins of attraction: $\mathcal{B}(G)$ converging to $(1,0,0)$ and
$\mathcal{B}(P)$ converging to $(0,1,0)$. The separatrix passes
through an unstable equilibrium $\xG^* \in (0,1)$ on the G--P edge
defined by:
\begin{equation}
  \alpha\Phi\bigl(\gamma + (1-\gamma)\xG^*\bigr)
  = (c_0 - \cP) + (\bP - \bG).
  \label{eq:tipping_point}
\end{equation}
\end{proposition}

Proofs are given in the Supporting Information. The point $\xG^*$ is
the tipping point: below it the population drifts toward the partial
adoption trap; above it toward full genuine adoption.

\begin{proposition}[Partial adoption trap]
\label{prop:trap}
For any initial condition $\boldsymbol{x}(0) \in \mathcal{B}(P)$,
the replicator dynamics converge to $(0,1,0)$. The basin
$\mathcal{B}(P)$ contains the region
$\{\boldsymbol{x} : \eeff(\boldsymbol{x}) < \estar\}$
for sufficiently small $\alpha$.
\end{proposition}

\begin{proposition}[Comparative statics on the separatrix]
\label{prop:comp_statics}
The tipping point $\xG^*$ is: (i) decreasing in $\alpha$; (ii)
decreasing in $B$; (iii) increasing in $c_0 - \cP$; (iv) increasing
in $\bP - \bG$; (v) non-monotone in $\gamma$.
\end{proposition}

These follow directly from implicit differentiation of
Eq~(\ref{eq:tipping_point}). Higher appropriability, larger systemic
benefit, smaller cost differential, and smaller private benefit gap all
expand the basin of genuine adoption. The non-monotonicity in $\gamma$
arises because partial adopters both help reach $\estar$ (reducing
$\xG^*$) and compete with genuine adopters in the replicator dynamics
(strengthening the trap).

\subsection*{Cost ratchet dynamics}

The time-varying cost $c(\eeff, t)$ from Eq~(\ref{eq:cost_decay}) is
now introduced, giving the full system evolving on
$S^2 \times [0, c_0]$:
\begin{align}
  \dot{\xG} &= \xG\bigl(\fG(\boldsymbol{x}, c, \alpha)
               - \fbar(\boldsymbol{x}, c, \alpha)\bigr), \nonumber \\
  \dot{\xP} &= \xP\bigl(\fP - \fbar(\boldsymbol{x}, c, \alpha)\bigr),
               \nonumber \\
  \dot{\xR} &= \xR\bigl(\fR - \fbar(\boldsymbol{x}, c, \alpha)\bigr),
               \nonumber \\
  \dot{c}   &= -\delta c \cdot
               \indicator_{\{\eeff(\boldsymbol{x}) > \estar\}}.
  \label{eq:full_system}
\end{align}

\begin{proposition}[Cost ratchet]
\label{prop:ratchet}
Every excursion above $\estar$ permanently reduces the cost of genuine
adoption, regardless of whether embedding is achieved. The system is
irreversible in the cost dimension: $c$ is non-increasing along any
trajectory.
\end{proposition}

\begin{proposition}[Critical excursion duration]
\label{prop:excursion}
There exists a minimum excursion duration $T^* > 0$ satisfying:
\begin{equation}
  \xG(T^*) = \xG^*(c_0 e^{-\delta T^*}),
  \label{eq:critical_T}
\end{equation}
such that excursions of duration $T > T^*$ lead to successful
embedding and excursions of duration $T < T^*$ lead to convergence
to the partial adoption trap, with the cost ratchet permanently
lowering the barrier for subsequent attempts.
\end{proposition}

Pilot programmes that fail to achieve permanent embedding nonetheless
lower future barriers through the cost ratchet. Pilots should therefore
be designed for \emph{duration above threshold} rather than speed of
rollout.

Four qualitatively distinct trajectory types arise:

\begin{description}
  \item[Type 1 (Direct trap):] Converges monotonically to $(0,1,0)$
    with no cost decay.
  \item[Type 2 (Failed crossing):] Approaches $\estar$ but does not
    sustain it for duration $T^*$; cost decays partially then relapses.
    Barrier is permanently lowered.
  \item[Type 3 (Successful embedding):] Stays above $\estar$ for
    $T > T^*$; system locks in at $(1,0,0)$ with $c \to 0$.
  \item[Type 4 (Cost ratchet oscillation):] Repeated crossings each
    decay cost a little; eventually a crossing becomes self-sustaining.
\end{description}

\subsection*{The trust game}

The appropriability parameter $\alpha$ is now endogenised. The
organisation announces sharing fraction $\hat\alpha$, doctors choose
strategies under $\alpha = \hat\alpha$, and if $\eeff > \estar$
systemic gains $V = B$ are realised. The organisation then chooses
actual sharing $\alpha_{\text{actual}} \in [0, \hat\alpha]$ to
maximise:
\begin{equation}
  U_{\text{org}} = (1 - \alpha_{\text{actual}})V
                 - \kappa(\hat\alpha - \alpha_{\text{actual}}),
  \label{eq:org_payoff}
\end{equation}
where $\kappa(\cdot)$ satisfies Assumption~\ref{ass:kappa}. The
first-order condition $\kappa'(\Delta) = V$ (where
$\Delta \equiv \hat\alpha - \alpha_{\text{actual}}^*$) gives:
\begin{equation}
  \frac{d\Delta}{dV} = \frac{1}{\kappa''(\Delta)} > 0,
  \label{eq:reneging_monotone}
\end{equation}
confirming that optimal reneging $\Delta$ is strictly increasing in
$V$; the organisation shares less as gains grow. Doctors update
beliefs according to:
\begin{equation}
  \dot\alpha = -\lambda\bigl(\alpha - \alpha_{\text{actual}}(V, \kappa)\bigr),
  \label{eq:belief_update}
\end{equation}
where $\lambda > 0$ is the belief updating rate.

\begin{proposition}[Trust sustainability]
\label{prop:trust}
In a repeated version of the trust game with discount factor
$\beta \in (0,1)$, the organisation sustains genuine commitment if
and only if:
\begin{equation}
  \beta \geq \beta^* \equiv \frac{V}{V + \kappa},
  \label{eq:trust_condition}
\end{equation}
where $\kappa > 0$ is the per-unit reputational cost under the local
linearisation $\kappa(\Delta) \approx \kappa\Delta$, standard in the
folk-theorem literature. $\beta^*$ is increasing in $V$ and decreasing
in $\kappa$.
\end{proposition}

\begin{proof}
The organisation defects if $V\Delta > \frac{\beta}{1-\beta}\kappa\Delta$.
Cancelling $\Delta > 0$ and rearranging: $V(1-\beta) > \beta\kappa$,
giving $\beta < V/(V+\kappa)$. Hence cooperation requires
$\beta \geq V/(V+\kappa)$. Note $\beta^* \in (0,1)$,
$\partial\beta^*/\partial V = \kappa/(V+\kappa)^2 > 0$, and
$\partial\beta^*/\partial\kappa = -V/(V+\kappa)^2 < 0$.
\end{proof}

\begin{proposition}[Trust-cost interaction]
\label{prop:trust_cost}
Repeated failed adoption attempts improve long-run prospects for genuine
adoption if and only if $\delta/\lambda > \theta^*$, where:
\begin{equation}
  \theta^* = \frac{|\partial x_G^*/\partial\alpha|}{\partial x_G^*/\partial c}
             \cdot \frac{\Delta_\alpha}{c_0} > 0.
  \label{eq:theta_star}
\end{equation}
Systems with fast cost decay and slow belief updating
($\delta/\lambda > \theta^*$) benefit from repeated attempts; systems
with slow cost decay and fast belief updating are worsened.
\end{proposition}

The trust game also introduces a self-confirming absorbing state: the
\emph{trust trap}, in which low beliefs about $\alpha$ prevent genuine
adoption, preventing the gains that would allow belief revision.
Escaping it requires making $\alpha$ contractually fixed and verifiable.

\subsection*{Coordination among doctors}

Direct payoff effects from the local strategy distribution are introduced
by adding coordination terms:
\begin{align}
  \fG^{\text{coord}} &= -c + \alpha\Phi(\eeff) + \bG
                        + \psi_G \xG - \psi_{\text{dev}} \xP,
  \label{eq:payoff_G_coord} \\
  \fP^{\text{coord}} &= -\cP + \bP - \psi_P \xG,
  \label{eq:payoff_P_coord}
\end{align}
where $\psi_G > 0$ captures peer genuine adoption benefit, $\psi_P > 0$
captures norm enforcement against partial adopters, and
$\psi_{\text{dev}} > 0$ captures the social cost to genuine adopters
of deviating in a partial adoption culture.

\begin{proposition}[Coordination amplification]
\label{prop:coord_amp}
Doctor-doctor coordination effects amplify bistability. Above the
separatrix, $\psi_G$ and $\psi_P$ accelerate convergence to $(1,0,0)$.
Below the separatrix, $\psi_{\text{dev}}$ accelerates convergence to
$(0,1,0)$. The separatrix is steeper and convergence faster on both
sides.
\end{proposition}

\begin{proposition}[Cultural lock-in]
\label{prop:cultural_lock}
In the presence of strong negative coordination ($\psi_{\text{dev}}$
large), the basin of attraction of genuine adoption contracts in both
the $\xG$ and $\xP$ dimensions. The partial adoption trap deepens as
it persists.
\end{proposition}

\begin{proposition}[Coordination-type interaction]
\label{prop:coord_type}
For point-solution AI ($\rho \to 1$), $\psi_G$, $\psi_P$,
$\psi_{\text{dev}} \to 0$: coordination effects vanish. For
system-change AI ($\rho \to 0$), coordination effects are strongest
precisely where they are most damaging.
\end{proposition}

\subsection*{Technology type index and the Value-Adoption Paradox}

A technology type index $\rho \in [0,1]$ is defined, where $\rho = 0$
is pure system-change AI and $\rho = 1$ is pure point-solution AI:
\begin{align}
  \estar(\rho) &= (1-\rho)\estar_0, &
  \alpha(\rho) &= \rho + (1-\rho)\alpha_0, \nonumber \\
  \bG(\rho) &= \bG^0 + \rho(\bP - \bG^0), &
  c_G(\rho) &= c_0(1-\rho).
  \label{eq:rho_parametrisation}
\end{align}

\begin{proposition}[Technology type threshold]
\label{prop:type_threshold}
There exists a critical technology type $\rho_c \in (0,1)$ such that
for $\rho < \rho_c$ the partial adoption trap is stable and the
dynamics are bistable, and for $\rho > \rho_c$ the partial adoption
trap is unstable and genuine adoption is the unique attractor:
\begin{equation}
  \rho_c = \frac{(\bP - \bG^0) + (c_0 - \cP)}{(\bP - \bG^0) + c_0}
         = 1 - \frac{\cP}{(\bP - \bG^0) + c_0}.
  \label{eq:rho_c}
\end{equation}
$\rho_c$ is increasing in $\bP$, decreasing in $\cP$, and increasing
in $c_0$.
\end{proposition}

\begin{proposition}[Value-adoption paradox]
\label{prop:paradox}
The technologies with the highest systemic value have the lowest $\rho$.
These are also the technologies that: (i) generate the deepest partial
adoption traps; (ii) create the greatest organisational reneging
temptation (highest $V$, requiring highest $\beta^*$); (iii) generate
the strongest negative coordination dynamics. Systemic value and
adoptability are negatively correlated across the technology type
spectrum.
\end{proposition}

\subsection*{Welfare analysis}

Assuming $B > n(c_0 - \bG)$ so that full genuine adoption is socially
optimal, the welfare loss per doctor from the partial adoption trap is:
\begin{equation}
  \Delta W = W_G - W_P
           = (\alpha B + \bG) - (\bP - \cP)
           = \alpha B + \bG - \bP + \cP,
  \label{eq:welfare_loss}
\end{equation}
which is increasing in $B$ and $\alpha$ and decreasing in $\rho$.

\begin{proposition}[Welfare loss monotonicity]
\label{prop:welfare}
The total social welfare loss $n\Delta W$ is largest for system-change
AI with the highest systemic value.
\end{proposition}

The welfare loss is invisible in standard adoption metrics: adoption
rates measure $\xG + \xP$, which is high at $(0,1,0)$. The correct
metric is effective adoption $\eeff = \xG + \gamma\xP$, which is low
at $(0,1,0)$ when $\gamma$ is small.

\section*{Discussion}
\label{sec:discussion}

The model identifies four distinct failure modes each requiring a
different class of intervention, with all four compounding each other
most severely for system-change AI.

\subsubsection*{Failure mode 1: The threshold failure}

The threshold failure arises because individual genuine adoption is not
individually rational below $\estar$. Two interventions address it.
Early genuine adopters receive a subsidy calibrated to make
$\fG \geq \fP$:
\begin{equation}
  s^* = (c_0 - \cP) + (\bP - \bG) - \alpha\Phi(\eeff),
  \label{eq:subsidy}
\end{equation}
which is self-terminating (reaching zero as $\eeff \to \estar$) and
needs to be paid only to the fraction of doctors needed to push $\xG$
above $\xG^*$. Additionally, seeding block adoption by clinical team
rather than dispersing early adopters allows local genuine adoption
frequencies to exceed the tipping point even when the global frequency
is below it.

\subsubsection*{Failure mode 2: The trust failure}

The trust failure arises because $\beta < \beta^*$. Three interventions
address it. First, multiyear implementation contracts with protected
budgets and leadership continuity requirements increase the effective
$\beta$. Second, mandatory reporting of how productivity gains are
distributed and contractual penalties for reneging increase $\kappa$
and hence reduce $\beta^*$. Third, and most robustly, making $\alpha$
contractually fixed and verifiable removes the trust game entirely.

\subsubsection*{Failure mode 3: The cost failure}

The cost failure arises when $\delta/\lambda < \theta^*$: trust erosion
outpaces cost decay. Investment in implementation support during the
embedding window (when $\eeff > \estar$) accelerates cost decay;
explicit implementation compacts anchor beliefs and slow erosion of
$\alpha$.

\subsubsection*{Failure mode 4: The cultural failure}

The cultural failure arises from strong negative coordination
($\psi_{\text{dev}}$ large). Making the unit of adoption the clinical
team rather than the individual doctor eliminates deviance costs by
construction. Cultural preparation through clinical champion programmes
before technology introduction reduces $\psi_{\text{dev}}$ before
dynamics begin.

\subsubsection*{Policy sequencing}

\begin{proposition}[Policy sequencing]
\label{prop:sequencing}
Interventions addressing different failure modes must be sequenced
correctly: (i) trust architecture must precede adoption requirements;
(ii) cultural preparation must precede individual incentives; (iii)
threshold subsidies must be concentrated in the seeding phase; (iv)
implementation support must be concentrated in the embedding window.
Interventions applied out of sequence will fail even if individually
well-designed.
\end{proposition}

Standard health system digital transformation policy introduces
technology with broad individual incentives addressing the threshold
failure superficially, without trust architecture or cultural
preparation, and with implementation support distributed evenly. The
model predicts this will systematically produce the partial adoption
trap: widespread partial adoption that generates process compliance
without pathway transformation, consistent with the accumulated
evidence on digital transformation in complex health systems
\cite{Greenhalgh2017}.

The required policy architecture is front-loaded (trust and culture
interventions before technology introduction), targeted (subsidies to
the seeding fraction rather than universal incentives), team-based
(adoption unit is the clinical team, not the individual), and
contractually structured (sharing rules fixed in advance and
verifiable).

\section*{Conclusion}
\label{sec:conclusion}

This paper develops an evolutionary game theoretic model of clinical
AI adoption that explains the systematic divergence between
point-solution and system-change AI diffusion. The model identifies
three compounding failure modes, threshold coordination failure,
organisational trust failure, and cultural lock-in, and shows they
are most severe for the technologies with the greatest systemic value:
the Value-Adoption Paradox. Standard adoption policy fails most where
it matters most.

The central formal contributions are: the characterisation of the
partial adoption trap as a stable equilibrium of the three-strategy
replicator dynamics; the cost ratchet result showing failed adoption
attempts permanently lower barriers; the trust-cost interaction showing
the ratchet benefit depends on the relative speed of cost decay and
belief erosion; the cultural lock-in result showing negative
coordination norms deepen the trap independently of individual
incentives; and the policy sequencing result showing order of
interventions is as important as the interventions themselves.

Several extensions would strengthen the analysis. Spatial or network
EGT would capture the role of professional networks in adoption
diffusion. Empirical calibration of the model parameters, particularly
$\estar$, $\delta$, and the coordination parameters, would allow
quantitative predictions. A natural empirical test is the predicted
negative correlation between systemic value ($\rho$ index) and
adoption depth (effective adoption $\eeff$ at steady state), which is
distinct from and uncorrelated with the adoption rate $\xG + \xP$.

Health systems require a fundamentally different policy framework for
system-change AI than for point solutions. The cost of applying
point-solution policy to system-change AI is not slow adoption but
permanent lock-in to partial adoption that looks like adoption without
generating systemic benefit.


\end{document}